\documentclass{llncs}
\usepackage{makeidx}  
\usepackage[cmex10]{amsmath}
\usepackage{stmaryrd} 
 \usepackage{amssymb} 
  \usepackage{color}
\usepackage{array}

\usepackage[pdftex]{hyperref}
\usepackage{url}

\usepackage{listings}
\usepackage{url}

\newdimen\proofrulebreadth \proofrulebreadth=.05em
\newdimen\proofdotseparation \proofdotseparation=1.25ex
\newdimen\proofrulebaseline \proofrulebaseline=2ex
\newcount\proofdotnumber \proofdotnumber=3
\let\then\relax
\def\hfi{\hskip0pt plus.0001fil}
\mathchardef\squigto="3A3B
\newif\ifinsideprooftree\insideprooftreefalse
\newif\ifonleftofproofrule\onleftofproofrulefalse
\newif\ifproofdots\proofdotsfalse
\newif\ifdoubleproof\doubleprooffalse
\let\wereinproofbit\relax
\newdimen\shortenproofleft
\newdimen\shortenproofright
\newdimen\proofbelowshift
\newbox\proofabove
\newbox\proofbelow
\newbox\proofrulename
\def\shiftproofbelow{\let\next\relax\afterassignment\setshiftproofbelow\dimen0 }
\def\shiftproofbelowneg{\def\next{\multiply\dimen0 by-1 }%
\afterassignment\setshiftproofbelow\dimen0 }
\def\setshiftproofbelow{\next\proofbelowshift=\dimen0 }
\def\setproofrulebreadth{\proofrulebreadth}

\def\prooftree{
\ifnum  \lastpenalty=1
\then   \unpenalty
\else   \onleftofproofrulefalse
\fi
\ifonleftofproofrule
\else   \ifinsideprooftree
        \then   \hskip.5em plus1fil
        \fi
\fi
\bgroup
\setbox\proofbelow=\hbox{}\setbox\proofrulename=\hbox{}%
\let\justifies\proofover\let\leadsto\proofoverdots\let\Justifies\proofoverdbl
\let\using\proofusing\let\[\prooftree
\ifinsideprooftree\let\]\endprooftree\fi
\proofdotsfalse\doubleprooffalse
\let\thickness\setproofrulebreadth
\let\shiftright\shiftproofbelow \let\shift\shiftproofbelow
\let\shiftleft\shiftproofbelowneg
\let\ifwasinsideprooftree\ifinsideprooftree
\insideprooftreetrue
\setbox\proofabove=\hbox\bgroup$\displaystyle 
\let\wereinproofbit\prooftree
\shortenproofleft=0pt \shortenproofright=0pt \proofbelowshift=0pt
\onleftofproofruletrue\penalty1
}

\def\eproofbit{
\ifx    \wereinproofbit\prooftree
\then   \ifcase \lastpenalty
        \then   \shortenproofright=0pt  
        \or     \unpenalty\hfil         
        \or     \unpenalty\unskip       
        \else   \shortenproofright=0pt  
        \fi
\fi
\global\dimen0=\shortenproofleft
\global\dimen1=\shortenproofright
\global\dimen2=\proofrulebreadth
\global\dimen3=\proofbelowshift
\global\dimen4=\proofdotseparation
\global\count255=\proofdotnumber
$\egroup  
\shortenproofleft=\dimen0
\shortenproofright=\dimen1
\proofrulebreadth=\dimen2
\proofbelowshift=\dimen3
\proofdotseparation=\dimen4
\proofdotnumber=\count255
}

\def\proofover{
\eproofbit 
\setbox\proofbelow=\hbox\bgroup 
\let\wereinproofbit\proofover
$\displaystyle
}%
\def\proofoverdbl{
\eproofbit 
\doubleprooftrue
\setbox\proofbelow=\hbox\bgroup 
\let\wereinproofbit\proofoverdbl
$\displaystyle
}%
\def\proofoverdots{
\eproofbit 
\proofdotstrue
\setbox\proofbelow=\hbox\bgroup 
\let\wereinproofbit\proofoverdots
$\displaystyle
}%
\def\proofusing{
\eproofbit 
\setbox\proofrulename=\hbox\bgroup 
\let\wereinproofbit\proofusing
\kern0.3em$
}

\def\endprooftree{
\eproofbit 
  \dimen5 =0pt
\dimen0=\wd\proofabove \advance\dimen0-\shortenproofleft
\advance\dimen0-\shortenproofright
\dimen1=.5\dimen0 \advance\dimen1-.5\wd\proofbelow
\dimen4=\dimen1
\advance\dimen1\proofbelowshift \advance\dimen4-\proofbelowshift
\ifdim  \dimen1<0pt
\then   \advance\shortenproofleft\dimen1
        \advance\dimen0-\dimen1
        \dimen1=0pt
        \ifdim  \shortenproofleft<0pt
        \then   \setbox\proofabove=\hbox{%
                        \kern-\shortenproofleft\unhbox\proofabove}%
                \shortenproofleft=0pt
        \fi
\fi
\ifdim  \dimen4<0pt
\then   \advance\shortenproofright\dimen4
        \advance\dimen0-\dimen4
        \dimen4=0pt
\fi
\ifdim  \shortenproofright<\wd\proofrulename
\then   \shortenproofright=\wd\proofrulename
\fi
\dimen2=\shortenproofleft \advance\dimen2 by\dimen1
\dimen3=\shortenproofright\advance\dimen3 by\dimen4
\ifproofdots
\then
        \dimen6=\shortenproofleft \advance\dimen6 .5\dimen0
        \setbox1=\vbox to\proofdotseparation{\vss\hbox{$\cdot$}\vss}%
        \setbox0=\hbox{%
                \advance\dimen6-.5\wd1
                \kern\dimen6
                $\vcenter to\proofdotnumber\proofdotseparation
                        {\leaders\box1\vfill}$%
                \unhbox\proofrulename}%
\else   \dimen6=\fontdimen22\the\textfont2 
        \dimen7=\dimen6
        \advance\dimen6by.5\proofrulebreadth
        \advance\dimen7by-.5\proofrulebreadth
        \setbox0=\hbox{%
                \kern\shortenproofleft
                \ifdoubleproof
                \then   \hbox to\dimen0{%
                        $\mathsurround0pt\mathord=\mkern-6mu%
                        \cleaders\hbox{$\mkern-2mu=\mkern-2mu$}\hfill
                        \mkern-6mu\mathord=$}%
                \else   \vrule height\dimen6 depth-\dimen7 width\dimen0
                \fi
                \unhbox\proofrulename}%
        \ht0=\dimen6 \dp0=-\dimen7
\fi
\let\doll\relax
\ifwasinsideprooftree
\then   \let\VBOX\vbox
\else   \ifmmode\else$\let\doll=$\fi
        \let\VBOX\vcenter
\fi
\VBOX   {\baselineskip\proofrulebaseline \lineskip.2ex
        \expandafter\lineskiplimit\ifproofdots0ex\else-0.6ex\fi
        \hbox   spread\dimen5   {\hfi\unhbox\proofabove\hfi}%
        \hbox{\box0}%
        \hbox   {\kern\dimen2 \box\proofbelow}}\doll%
\global\dimen2=\dimen2
\global\dimen3=\dimen3
\egroup 
\ifonleftofproofrule
\then   \shortenproofleft=\dimen2
\fi
\shortenproofright=\dimen3
\onleftofproofrulefalse
\ifinsideprooftree
\then   \hskip.5em plus 1fil \penalty2
\fi
}

\newcommand{\infer}[2]
     {\prooftree
          #1 
          \justifies #2
      \endprooftree}

\newcommand{\ie}{i.e.\ }

\newcommand{\sch}{\mathcal{S}}  

\newcommand{\restrict}[1]{#1\!\downarrow\!1 }
\newcommand{\config}{\mathit{Config}}

\newcommand{\ens}[1]{\mathbb{#1}}
\newcommand{\dom}{\textit{dom}}

\newcommand{\expr}[1]{\mathit{\MakeUppercase #1}}
\newcommand{\ea}{\expr{E}}

\newcommand{\eaa}{\ea_1}

\newcommand{\ean}{\ea_n}
\newcommand{\FV}{\mathcal{V}}
\newcommand{\Var}{\ens{V}}
\newcommand{\variable}[1]{\mathit{\MakeUppercase #1}}
\newcommand{\xa}{\variable{x}}
\newcommand{\xaa}{\variable{x}_1}
\newcommand{\xai}{\variable{x}_i}
\newcommand{\xan}{\variable{x}_n}
\newcommand{\xb}{\variable{y}}
\newcommand{\xc}{\variable{z}}
\newcommand{\xd}{\variable{u}}
\newcommand{\xe}{\variable{r}}
\newcommand{\xf}{\variable{c}}

\newcommand{\threada}{\mathit{x}}
\newcommand{\threadb}{\mathit{y}}

\newcommand{\Dom}{\ens{W}}

\newcommand{\elmt}[1]{\mathit{#1}}
\newcommand{\da}{\elmt{d}}
\newcommand{\daa}{\elmt{d}_1}
\newcommand{\dai}{\elmt{d}_i}
\newcommand{\dan}{\elmt{d}_n}
\newcommand{\du}{\elmt{u}}
\newcommand{\dv}{\elmt{v}}
\newcommand{\dw}{\elmt{w}}

\newcommand{\db}{\elmt{b}}
\newcommand{\Operator}{\ens{O}}
\newcommand{\oper}[1]{\mathit{#1}}
\newcommand{\op}{\oper{op}} 
\newcommand{\cst}[1]{c_{#1}} 
\newcommand{\msuc}[1]{\oper{suc}_{#1}}
\newcommand{\mpred}{\oper{pred}}
\newcommand{\meq}[1]{\oper{eq}_{#1}}

\newcommand{\Command}[1]{\mathit{\MakeUppercase #1}}
\newcommand{\ca}{\Command{C}}

\newcommand{\cb}{\Command{C'}}
\newcommand{\cc}{\Command{D}}
\newcommand{\cd}{\Command{D'}}

\newcommand{\rgl}{::=}

\newcommand{\instr}[1]{\mathtt{#1}}

\newcommand{\iasg}{\instr{:=}}
\newcommand{\isep}{\instr{\ ;\ }}
\newcommand{\iwh}{\instr{while}}
\newcommand{\iwhile}{\instr{while}}
\newcommand{\iif}{\instr{if\ }}
\newcommand{\ithen}{\instr{\ then\ }}
\newcommand{\ielse}{\instr{\ else\ }}

\newcommand{\iskip}{\instr{skip}}

\newcommand{\typenv}{\Gamma}
\newcommand{\typop}{\Delta}
\newcommand{\Imp}{\vDash}
\newcommand{\imp}{\vdash}

\newcommand{\store}{\mu}
\newcommand{\storb}{\sigma}

\newcommand{\Cmpt}[1]{\Rightarrow_{#1}}
\newcommand{\Csmt}[1]{\stackrel{\texttt{\,#1}}{\to}}

\newcommand{\ptime}{ \mathrm{FPtime}}

\newcommand{\Temps}[1]{\textit{Time}_{#1}}
\newcommand{\temps}[2]{\Temps{#1}(#2)}

\newcommand{\slat}[1]{\ty{#1}}
\newcommand{\sla}{\slat{\alpha}}
\newcommand{\slb}{\slat{\beta}}
\newcommand{\SL}{\{\tiera,\tierb\}}
\newcommand{\ord}{\preceq}

\newcommand{\meet}{\wedge}
\newcommand{\join}{\vee}

\newcommand{\tier}[1]{\mathbf{#1}}
\newcommand{\tiera}{\tier{0}}
\newcommand{\tierb}{\tier{1}}

\newcommand{\vrai}{\texttt{tt}}
\newcommand{\faux}{\texttt{ff}}

\newcommand{\ty}[1]{#1}
\newcommand{\tya}{\ty{\alpha}}
\newcommand{\tyaa}{\ty{\alpha_1}}

\newcommand{\tyan}{\ty{\alpha_n}}
\newcommand{\tyb}{\ty{\beta}}

\newcommand{\sra}{\rightarrow}

\newcommand{\motVide}{\oper{\epsilon}}

 \newcommand{\envprog}[1]{\texttt{#1}}

 \newcommand{\gauche}{\envprog{Left}}
 \newcommand{\droite}{\envprog{Right}}
 
 \newcommand{\vetat}{\envprog{State}}

\newcommand{\dord}{\unlhd} 

\newcommand{\inuniv}[1]{\instr{exec}}

\newcommand{\inspec}[1]{\instr{spec}}

\newcommand{\sem}[1]{\llbracket #1 \rrbracket}

\newcommand{\taille}[1]{|#1|} 

\begin{document}
%
%
\pagestyle{headings}  

\title{Complexity Information Flow in a Multi-threaded Imperative Language}
\author{Jean-Yves Marion \and Romain P\'echoux}
\institute{Universit\'e de Lorraine, CNRS and INRIA \\
   LORIA \\
\email{jean-yves.marion@loria.fr,romain.pechoux@loria.fr}
}

\maketitle

\begin{abstract}
We propose a type system to analyze the time consumed by multi-threaded imperative programs with a shared global memory, which delineates a class of safe multi-threaded programs.
 We demonstrate that a safe multi-threaded program runs in polynomial time if 
 (i)  it is strongly terminating wrt a non-deterministic scheduling policy or (ii) it terminates wrt a deterministic and quiet scheduling policy. 
 As a consequence, we also characterize the set of polynomial time functions.
The type system presented is based on the fundamental notion of data tiering, which is central in implicit computational complexity. It regulates the information flow in a computation. This aspect is  interesting in that the type system bears a resemblance to typed based information flow analysis and notions of non-interference.
As far as we know, this is the first characterization by a type system of polynomial time multi-threaded programs.
\end{abstract}

\section{Introduction}
The objective of this paper is to study the notion of complexity flow analysis introduced in~\cite{MarionLICS} in the setting of concurrency.  Our model of concurrency is a simple multi-threaded imperative programming language where threads communicate through global shared variables.
The measure of time complexity that we consider for multi-threaded programs is the processing time. That is the total time for all threads to complete their tasks. As a result, the time measure gives an upper bound on the number of scheduling rounds.
The first contribution of this paper is a novel type system, which guarantees that each strongly terminating safe multi-threaded program runs in polynomial time  (See Section~\ref{sec:safeprog} and Theorem~\ref{thm:Main}). Moreover, the runtime upper bound holds for all thread interactions.
As a simple example, consider the two-thread program:
\begin{center}
\begin{tabular}{lcl}
\begin{lstlisting}
$\threada:\ \iwhile (\xa^{\tierb}==\xb^{\tierb}) \{ \iskip \}$
   	${ \ca; }$
	$\xa^{\tierb} \iasg \neg \xa^{\tierb}$
\end{lstlisting}
	& \hspace{2cm} &
\begin{lstlisting}
$\threadb: \ \iwhile (\xa^{\tierb} \neq \xb^{\tierb}) \{ \iskip \}$
   	${ \ca'; }$
	$\xb^{\tierb} \iasg \neg \xb^{\tierb}$
\end{lstlisting}
\end{tabular}
\end{center}

This example  illustrates a simple synchronization protocol between two threads $\threada$ and $\threadb$.  Commands $\ca$ and $\ca'$ are critical sections, which are assumed not to modify $\xa$ and $\xb$. The operator $\neg$ denotes the boolean negation. 
Both threads are safe if commands $\ca$ and $\ca'$ are safe with respect to the same typing environment.
Our first result states that this  two-thread program runs in polynomial time (in the size of the initial shared variable values) if it is strongly terminating and safe.

Then, we consider a class of deterministic schedulers, that we call quiet (see Section~\ref{sec:ds}). 
The class of deterministic and quiet schedulers contains all deterministic scheduling policies  which depend only on threads.  
A typical example is a round-robin scheduler. 
The last contribution of this paper is that a safe multi-threaded program which is terminating wrt to a deterministic and quiet scheduler, runs in polynomial time.
 Despite the fact it is not strongly terminating, the two-thread program below terminates under a round-robin scheduler, if $\ca$ and $\cb$ terminate.
 \begin{center}
\begin{tabular}{lcl}
\begin{lstlisting}
   $\threada: \ \iwhile ( \xa^{\tierb}> 0)$
       $\{\ca;$
       $\xc^\tierb \iasg 0 : \tierb\} : {\tierb}$ 
 \end{lstlisting}
	& \hspace{2cm} &
\begin{lstlisting}
  $\threadb:\ \iwhile ( \xc^{\tierb}> 0)$
      $\{\ca'; $
       $\xa^\tierb \iasg 0 : \tierb \} : {\tierb}   $
\end{lstlisting}
\end{tabular}
\end{center}
If commands $\ca$ and $\cb$ are safe, then this two-thread program runs in polynomial time wrt to a round-robin scheduler.
The last contribution is that if we just consider one-thread programs, then we characterize exactly $\ptime$, which is the class of polynomial time functions. (See Theorem~\ref{erst})

The first rational behind our type system comes from data-ramification concept of Bellantoni and Cook~\cite{BC92} and Leivant~\cite{Lei-predicativeI}. The type system has two atomic types $\tiera$ and $\tierb$ that we called tiers. The type system precludes that values flow from tier $\tiera$ to tier $\tierb$ variables.  Therefore, it prevents circular algorithmic definitions, which may possibly lead to an exponential length computation. More precisely, explicit flow from $\tiera$ to $\tierb$ is forbidden by requiring that the type level of the assigned variable is less or equal than the type level of the source expression. Implicit flow is prevented by requiring that (i) branches of a conditional are of the same type and (ii) guard and body of while loops are of tier $\tierb$. If we compare with data-ramification concept of~\cite{BC92,Lei-predicativeI}, tier $\tierb$ parameters correspond to variables on which a ramified recursion is performed whereas tier $\tiera$ parameters correspond to variables on which recursion is forbidden. 

The second rational behind our type system comes from secure flow analysis. See Sabelfeld and Myers survey~\cite{SabelfeldMyersJSAC} to have an overview on information flow analysis. In~\cite{VolpanoIS96} for sequential imperative programs and in~\cite{SV98} for multi-threaded imperative programming language, Irvine,  Smith and Volpano give a type system to certify a confidentiality policy. Types are based on security levels say H (High) and L (Low). The type system prevents that there is no leak of information from level H to level L, which is similar to our type system: $\tiera$ (resp. $\tierb$) corresponds to H (resp. L). In fact, our approach rather coincides with an integrity policy~\cite{Biba77} (i.e "no read down" rule) than with a confidentiality one~\cite{BellLapadula76}.
A key property is the non-interference, which says that values of level L don't changed values of level H. We demonstrate a similar non-interference result which states that values stored in tier $\tierb$ variables are independent from tier $\tiera$ variables. See Section~\ref{sec:NI} for a precise statement.
From this, we demonstrate a temporal non-interference properties which expresses that the number of unfolded (i.e. the length) while loops only depends on tier $\tierb$ variables, see Section~\ref{sec:TNI}. The temporal non-interference property is the crucial point to establish complexity bounds.

From a practical standpoint, an important issue is the expressivity of the class of safe multi-threaded programs. 
With this work and~\cite{MarionLICS}, we introduce a new approach in implicit computational complexity based on a type system. This study focuses on the intrinsic mechanisms which lead to analyze computational complexity.
This approach seems promising because it treats common algorithmic control structures like while-loops as well as sequential and parallel composition. Several examples are presented in Appendix.

\noindent\textbf{Related works.}
An important source of inspiration comes from Implicit Computational Complexity (ICC). Beside the works of Bellantoni, Cook and Leivant already cited, there are works on light logics~\cite{Girard98,Baillot}, on linear types~\cite{Hofmann03}, and interpretation methods~\cite{BMM11,MP09}, just to mention a few. There are also works on resource control of imperative language like~\cite{Jones01,JonesK09,NigglW06}. 
Only a few studies based on ICC methods are related to resource control of concurrent computational models. 
In~\cite{ADZ04}, a bound on the resource needed by synchronous cooperative threads in a functional framework is computed. The paper~\cite{AD07} provides a static analysis for ensuring feasible reactivity in a synchronous $\pi$-calculus. In~\cite{MA11} an elementary affine logic is introduced to tame the complexity of a modal call-by-value lambda calculus with multi-threading and side effects. There are also works on the termination of multi-threaded imperative languages~\cite{CPR07}. In this paper, we separate complexity analysis from termination analysis but the tools on termination can be combined with our results since most of them require strong normalization of the considered process as an assumption.
Finally our type system in this paper may be seen as a simplification of the type system of~\cite{MarionLICS} for imperative language but in return there is no declassification mechanism. 


\section{A complexity flow type system} \label{sec:CFTS}
\subsection{A multi-threaded programming language} \label{subSynSem}
We introduce a multi-threaded imperative programming language similar to the language of~\cite{SV98,Boudol} and which is an extension of the simple while-imperative programming language of~\cite{JonesCC}. A multi-threaded program consists in a finite set of threads where each thread is a while-program. Threads run concurrently on a common shared memory. A thread interacts with other threads by reading and writing on the shared memory. 

Commands and expressions are built from a set $\Var$ of variables, and a set $\Operator$ of operators of fixed arity including constants (operators of arity $0$) as follows:
$$
 \begin{array}{llll}
\textit{Expressions} \qquad \qquad & \eaa,\ldots,\ean  
  & \rgl &\xa 
   \ |\ \op(\eaa,\ldots,\ean)   \qquad \qquad \qquad  \xa \in \Var, \op  \in \Operator \\
\textit{Commands} &  \ca, \cb  
    & \rgl & \xa \iasg \ea \ | \  \ca \isep \cb \ |\   \iskip \ |\ \iif \ea \ithen \ca \ielse \cb \\
& & & \ |\ \iwh(\ea)\{ \ca \}  
 \end{array}
$$

A multi-threaded program $M$ (or just program when there is no ambiguity) is a finite map 
from thread identifiers $\threada, \threadb,\ldots$ to commands. We write $\dom(M)$ to denote the set of thread identifiers. 
 Note also that we do not consider the ability of generating new threads.
Let $\FV(I)$ be the set of variables occurring in $I$, where $I$ is an expression,  a command or a multi-threaded program.

\subsection{Semantics}
We give a standard small step operational semantics for multi-threaded programs. 
Let $\Dom$ be the set of words\footnote{Our result could be generalized to other domains such as binary trees or lists. However we have restricted this study to words in order to lighten our presentation.} over a finite alphabet  $\Sigma$ including two words $\vrai$ and $\faux$ that denote true and false. The length of a word $\da$ is denoted  $\taille{\da}$. A store $\store$ is a finite mapping from $\Var$ to $\Dom$.  
We write $\store[\xaa \leftarrow \daa,\ldots,\xan \leftarrow \dan]$ to mean the store $\store'$ where $\xai$ is updated to $\dai$.

The evaluation rules for expressions and commands are given in  Figure~\ref{fig:Com}.  
Each operator of arity $n$ is interpreted by a total function $\sem{\op}:\Dom^n \mapsto \Dom$. 
The judgment $\store \Imp \ea \Csmt{e} \da$ means that the expression $\ea$ is evaluated to the word $\da \in \Dom$ wrt $\store$.
A configuration $c$ is either a pair of store and command, $\store \Imp \ca$, or a store $\store$.
The judgment $\store \Imp \ca \Csmt{s} \store'$ 
expresses that $\ca$ terminates and outputs the store $\store'$.
 $\store \Imp \ca \Csmt{s} \store' \Imp \ca'$ means that the evaluation of $\ca$ is still in progress: 
the command has evolved to $\ca'$ and the store has been updated to $\store'$. 

For a multi-threaded program $M$, the store $\store$ plays the role of a global memory shared by all threads.
The store $\store$ is the only way for threads to communicate. The definition of  the global relation $\Csmt{g}$ is given in Figure~\ref{fig:Com}, where $M\!-\!\threada$ is the restriction of $M$ to $dom(M)\!-\!\{\threada\}$ 
and $ M[\threada:=\ca_1]$ is the map $M$ where the command assigned to $\threada$ is updated to $\ca_1$.  
At each step, a thread $\threada$ is chosen non-deterministically. Then, one step of $\threada$ is performed and the control returns to the upper level.  Note that the rule  \textit{(Stop)} halts the computation of a thread. In what follow, let $\emptyset$ be a notation for the (empty) multi-threaded program (i.e. all threads have terminated).
We will discuss of deterministic scheduling policy in the last section. 
\begin{figure*}[t]
\hrulefill\\
$$
\begin{array}{c}
\infer{\phantom{ \Imp \xa \iasg \ea \Csmt{s} }}{\store \Imp \xa \Csmt{e} \store(\xa)}  
\quad \quad \infer{\store \Imp \ea_1 \Csmt{e} \da_1 \quad \ldots \quad \store \Imp \ea_n \Csmt{e} \da_n } 
{\store \Imp \op(\ea_1,\ldots,\ea_n) \Csmt{e} \sem{\op}(\da_1,\ldots,\da_n)} 
\\[15pt]
\infer{\phantom{ \Imp \xa \iasg \ea \Csmt{s} } }
{\store \Imp \iskip \Csmt{s} \store} 
\quad \quad \infer{\store \Imp \ea \Csmt{e} \da }
{\store \Imp \xa \iasg \ea \Csmt{s} \store[\xa \leftarrow \da]}
\quad \quad \infer{\store \Imp \ca_1 \Csmt{s} \store_1} 
{\store \Imp \ca_1 \isep \ca_2 \Csmt{s} \store_1 \Imp \ca_2} 
\\[15pt]
\infer{\store \Imp \ca_1 \Csmt{s} \store_1 \Imp \ca_1' } 
{\store \Imp \ca_1 \isep \ca_2 \Csmt{s} \store_1 \Imp \ca_1' ; \ca_2} 
\quad \quad \infer{\store \Imp \ea \Csmt{e} w,\ w\in \{\vrai,\faux\}} 
{\store \Imp \iif \ea \ithen \ca_{\vrai} \ielse \ca_{\faux} \ \Csmt{s}\ \store \Imp \ca_{w}} 
 \\[15pt]
\infer{\store \Imp \ea \Csmt{e} \faux} 
{\store \Imp \iwhile (\ea) \{\ca\} \ \Csmt{s}\ \store}
\quad \quad \infer{\store \Imp \ea \Csmt{e} \vrai }
{\store \Imp  \iwhile (\ea) \{\ca\} \ \Csmt{s}\ \store \Imp \ca ;  \iwhile (\ea) \{\ca\}}\textit{(W}_\vrai\textit{)}
\\[15pt]
\infer{M(\threada)= \ca \quad \store \Imp \ca \Csmt{s} \store_1  }
{\store \Imp  M \ \Csmt{g}\ \store_1 \Imp M-\threada}\ \textit{(Stop)}
\quad \quad \infer{M(\threada)= \ca \quad \store \Imp \ca \Csmt{s} \store_1 \Imp \ca_1 }
{\store \Imp  M \ \Csmt{g}\ \store_1 \Imp M[\threada:=\ca_1]}\ \textit{(Step)}
\end{array}
$$
\caption{Small step semantics of expressions, commands and multi-threads}
\label{fig:Com}
\hrulefill
\end{figure*}

 A multi-threaded program $M$ is strongly terminating, noted $M\!\Downarrow$, 
if for any store, all reduction sequences starting from $M$ are finite.
Let  $\Csmt{h}^k$ be the $k$-fold self composition and $\Csmt{h}^*$ be the reflexive and transitive closure of the relation $\Csmt{h}$, $\texttt{h} \in \{\texttt{s,g}\}$.
The running time of a strongly terminating program $M$ is the function $\Temps{M}$ from $\Dom^n$ to $\mathbb{N}$
defined by:
\begin{gather*}
\temps{M}{\da_1,\ldots,\da_n} = \max \{ k \ | \  \store_0[\xaa \leftarrow \daa,\ldots,\xan \leftarrow \dan] \Imp M \Csmt{g}^k \store \Imp \emptyset\}
\end{gather*}
where  $\store_0$ is the empty store that maps each variable to the empty word $\epsilon \in \Dom$.

\noindent
A strongly terminating multi-threaded program $M$ is running in polynomial time if 
there is a polynomial $Q$ such that for all $\da_1,\ldots,\da_n \in \Dom$,
$\temps{M}{\da_1,\ldots,\da_n} \leq Q(\max_{i=1,n}\taille{\da_i})$.
Observe that, in the above definition, the time consumption of an operator is considered as constant, which is fair if operators are supposed to be computable in polynomial time.


\subsection{Type system}\label{sec:TypeSystem}
Atomic types are elements of the boolean lattice $(\{\tiera,\tierb\},\ord,\tiera,\join,\meet)$ where  $\tiera \ord \tierb$.
We call them \emph{tiers} accordingly to the data ramification principle of~\cite{Leivant94}. 
We use  $\sla,\slb,\ldots$ for tiers.
 A variable typing environment $\typenv$ is a finite mapping from $\Var$ to $\SL$, which assigns a single tier to each variable.
An operator typing environment $\typop$ is a mapping that associates to each operator $\op$ a set of operator types $\typop(\op)$, where the operator types corresponding to an operator of arity $n$ are of the shape $\sla_1 \to \ldots \sla_n \to \sla$ with $\sla_i,\sla \in \SL$ using implicit right associativity of $\sra$.
We write $\dom(\typenv)$ (resp. $\dom(\typop)$) to denote the set of variables typed by $\typenv$ (resp. the set of operators typed by $\typop$).
Figure~\ref{fig:TypeCom} gives the typing discipline for expressions, commands and multi-threaded programs.
Given a multi-threaded program $M$, a variable typing environment $\typenv$ and an operator typing environment $\typop$, 
$M$ is \emph{well-typed}
if for every $\threada \in dom(M)$, $\typenv, \typop \imp M(\threada) : \tya$ for some tier $\tya$.

\begin{figure}[t]
\hrulefill\\
$$
\begin{array}{c}
\infer{ \typenv(\xa)=\tya}{\typenv,\typop \imp \xa:\tya}
 \quad  \quad \infer{\typenv, \typop \imp \xa: \tyb  \qquad \typenv, \typop \imp \ea:\tya} 
{\typenv, \typop \imp \xa \iasg \ea:\tyb}
{\tyb \ord \tya}
\\[15pt]
\infer{\typenv,\typop \imp \eaa:\tyaa \ldots  \typenv,\typop \imp \ean:\tyan \qquad {\tyaa} \sra \ldots \sra {\tyan} \sra {\tya} \in \typop(\op)}
{\typenv,\typop \imp \op(\eaa,\ldots,\ean):\tya }
\\[15pt]
\infer{\typenv, \typop \imp \ea:\tierb \qquad \typenv, \typop \imp \ca:\tya} 
{\typenv, \typop \imp \iwhile (\ea) \{\ca\}:\tierb}
\quad 
 \quad \infer{\typenv, \typop \imp \ca:\tya \qquad \typenv, \typop \imp \cb:\tyb} 
{\typenv, \typop \imp \ca \isep \cb: \tya \join \tyb}
\\[15pt]
\infer{} 
{\typenv, \typop \imp \iskip:\tya}
\quad \quad  \infer{\typenv, \typop \imp \ea:\tya \qquad \typenv, \typop \imp \ca: \tya \qquad \typenv, \typop \imp \cb:\tya} 
{\typenv, \typop \imp \iif \ea \ithen \ca \ielse \cb:\tya}
\end{array}
$$
\caption{Type system for expressions, commands\label{fig:TypeCom}}
\hrulefill
\end{figure}

Notice that the subject reduction property is not valid, because we don't explicitly have any subtyping rule.
However, a weak subject reduction property holds: If $\store \Imp \ca \Csmt{s} \store' \Imp \ca'$ then $\typenv,\typop \imp \ca':\tyb$ where $\tyb \ord \tya$.

\section{Safe multi-threaded program} \label{hector}
\subsection{Neutral and positive operators} 
As in \cite{MarionLICS}, we define two classes of operators called neutral and positive.
For this, let $\dord$ be the sub-word relation over $\Dom$, which is defined by $\dv \dord \dw$, iff there are $\du$ and $\du'$ such that $\dw = \du.\dv.\du'$, where $.$ is the concatenation.

\noindent
An operator $\op$ is \emph{neutral} if: 
\begin{enumerate}
\item either $\sem{\op}:\Dom \to \{{\vrai},{\faux}\}$ is a predicate;
\item or  for all $\da_1,\ldots,\da_n \in \Dom,\ \exists i \in \{1,\ldots,n\}$,
$
\sem{\op}(\da_1,\ldots,\da_n) \dord {\da_i}
$.
\end{enumerate}
An operator $\op$  is \emph{positive} if there is a constant $\cst{op}$ such that:
\begin{align*}
\taille{\sem{\op}(\da_1,\ldots,\da_n)} & \leq \max_i \taille{\da_i} + \cst{op}
\end{align*}

A neutral operator is always a positive operator but the converse is not true.
In the remainder, we assume that operators are all neutral or positive.

\subsection{Safe environments and safe multi-threaded programs}\label{sec:safeprog}
An operator typing environment $\typop$ is \emph{safe} if for each $ \op \in \dom(\typop)$ of arity $n$ 
and for each  $\tyaa \sra \ldots \sra \tyan \sra \alpha \in \typop(\op)$, we have $\alpha \ord \meet_{i=1,n} \tya_i$,
and if the operator $\op$ is positive but  not neutral, then $\tya=\tiera$.

Now, given $\typenv$ a variable typing environment and $\typop$ a operator typing environment,
we say that $M$ is a \emph{safe multi-threaded program}  if $M$ is well-typed wrt $\typenv$ and $\typop$
and $\typop$ is safe.

Intuitively, a tier $\tiera$ argument is unsafe. This means that it cannot be used as a loop guard. 
So for "loop-safety" reasons, if an operator has a tier $\tiera$ argument then the result is necessarily of tier $\tiera$. 
In return, a positive operator  can increase the size of its arguments.
On the other hand, a neutral operator does not increase the size of its arguments. So, we can apply it safely everywhere.
The combination of the type system, which guarantees some safety properties on the information flow, and of operator specificities 
provides time bounds.

\begin{example}
Given a word $\da$, the operator $\meq{\da}$ tests whether or not its agument begins with the prefix $\da$ and $\mpred{}$ computes the predecessor.
$$
\begin{array}{lll}
\sem{\meq{\da}}(\du) = 
   \begin{cases}
     =\vrai &\text{if }\du=\da \dw  \\
        = \faux &\text{otherwise} 
   \end{cases}                  
& \qquad \qquad \qquad&
\sem{\mpred}(\du)   =
     \begin{cases}
      = \epsilon  &\text{if }\du=\epsilon \\ 
      = \dw  &\text{if\ } \du=\ell.\dw,\ \ell \in \Sigma \\
      \end{cases}    
\end{array}
$$
Both operators are neutral. This means that their types satisfy
$\typop(\mpred),\typop(\meq{u}) \subseteq \{\tiera \sra \tiera, \tierb \sra \tierb, \tierb \sra \tiera \}$
wrt to a safe environment $\typop$. 
The operator $\msuc{\da}$ adds a prefix $\da$. It is positive, but not neutral.
So,  $\typop(\msuc{\da}) \subseteq \{\tierb \sra \tiera , \tiera \to \tiera \}$:
$$
\begin{array}{llllll}
\textit{(Positive)} & \sem{\msuc{\da}}(\db)&=\da.\db && \da \in \Sigma
\end{array}
$$
\end{example}

\section{Sequential and concurrent non-interferences}\label{sec:NI}
In this section, we demonstrate that classical non-interference results are obtained through the use of the considered type system. For that purpose, we introduce some intermediate lemmata.
The confinement Lemma expresses the fact that no tier $\tierb$ variables are modified by a command of tier $\tiera$.

\begin{lemma}[Confinement]\label{lem:confinement}
Let $\typenv$ be a variable typing environment and $\typop$ be a safe operator typing environment. 
If $\typenv, \typop \imp \ca:\tiera$,
then every variable assigned to in $\ca$ is of type $\tiera$,
and $\ca$ does not contain while loops.
\end{lemma}  

\begin{proof}
By induction on the structure of $\ca$.
\qed
\end{proof}

The following lemma, called simple security, says that only variables at level $\tierb$ will have their content read in order to evaluate an expression $\ea$ of type $\tierb$.

\begin{lemma}[Simple security]\label{lem:simpSec}
Let $\typenv$ be a variable typing environment and $\typop$ be a safe operator typing environment. 
If $\typenv, \typop \imp \ea:\tierb$,
then for every $\xa \in \FV(\ea)$, we have $\typenv(\xa) = \tierb$.
Moreover, all operators in $\ea$ are neutral. 
\end{lemma}  

\begin{proof}
By induction on $\ea$, and using the fact that $\ea$ is necessarily only composed of operators of type $\tierb \sra \ldots \sra \tierb \sra \tierb$,
because  the environment is safe.
\qed
\end{proof}

\begin{definition}
Let $\typenv$ be a variable typing environment and $\typop$ be an operator typing environment. 
\begin{itemize}
\item The equivalence relation $\approx_{\typenv,\typop}$ on stores is defined as follows: \\
 $\store \approx_{\typenv,\typop} \storb$ iff for every $\xa \in \dom(\Gamma)$ s.t. $\typenv(\xa) = \tierb$ we have $\store(\xa)=\storb(\xa)$
\item The relation $\approx_{\typenv,\typop}$ is extended to commands as follows:
\begin{enumerate}
\item If $\ca=\cb$ then $\ca \approx_{\typenv,\typop} \cb$
\item If $\typenv, \typop \imp \ca:\tiera$ and $\typenv, \typop \imp \cb:\tiera$ then $\ca \approx_{\typenv,\typop} \cb$
\item If $\ca \approx_{\typenv,\typop} \cb$ and $\cc \approx_{\typenv,\typop} \cd$ then $\ca;\cc\approx_{\typenv,\typop}  \cb ;\cd$ 
\end{enumerate}
\item Finally, it is extended to configurations as follows: \\
 If $\ca \approx_{\typenv,\typop} \cb$ and $\store \approx_{\typenv,\typop} \storb$ then $\store \Imp \ca   \approx_{\typenv,\typop} \storb \Imp \cb$
\end{itemize}
\end{definition}

\begin{remark}\label{rem:SimpleSecurity}
A consequence of Lemma~\ref{lem:simpSec} is that if $\store \approx_{\typenv,\typop} \storb$  and if $\typenv, \typop \imp \ea:\tierb$,
then computations of $\ea$ are identical under the stores $\store$ and $\storb$ , that is $\store \Imp \ea \Csmt{e} \da$ and $\storb \Imp \ea \Csmt{e} \da$.
\end{remark}

We now establish a sequential non-interference Theorem which states that if $\xa$ is variable of tier $\tierb$ then the value stored in $\xa$ is independent from variables of tier $\tiera$.

\begin{theorem}[Sequential non-interference] \label{thm:soundness}
Assume that $\typenv$ is a variable typing environment and $\typop$ is a safe operator typing environment
s.t.  $\typenv,\typop \imp \ca:\tya$ and $\typenv,\typop \imp \cc:\tya$. 
Assume also that $\store \Imp \ca \approx_{\typenv,\typop} \storb \Imp \cc$. Then, we have:
\begin{itemize}
\item  if $\store \Imp \ca \Csmt{s} \store' \Imp \ca'$ then there exists $\storb'$  and $\cc'$ 
such that $\storb \Imp \cc \Csmt{s}^* \storb' \Imp \cc'$ 
 and $\store' \Imp \ca' \approx_{\typenv,\typop} \storb' \Imp \cc'$,
\item if $\store \Imp \ca \Csmt{s} \store'$ then there exists $\storb'$  such that $\storb \Imp \cc \Csmt{s}^* \storb'$ 
and $\store' \approx_{\typenv,\typop} \storb'$
\end{itemize}
\end{theorem}

\begin{proof}
First suppose that $\alpha = \tiera$. 
Confinement Lemma~\ref{lem:confinement} implies that $\store' \approx_{\typenv,\typop} \storb'$ since no tier $\tierb$ variable is changed. 
Second suppose that $\alpha=\tierb$. We proceed by induction on $\ca$.
Suppose that $\ca$ is $\iwhile (\ea) \{\ca_1\}$ and the evaluation under $\store$ is:
$$
\infer{\store \Imp \ea \Csmt{e} \vrai }
{\store \Imp  \iwhile (\ea) \{\ca_1\} \ \Csmt{s}\ \store \Imp \ca_1 ;  \iwhile (\ea) \{\ca_1\}}\quad \textit{(W}_{\vrai}\textit{)}
$$
By Remark~\ref{rem:SimpleSecurity}, the evaluation of $\ea$ under $\storb$ is necessarily $\vrai$. 
Since $\ca$ is an atomic command, $\ca \approx_{\typenv,\typop}  \cc$ implies $\ca=\cc$. 
As a result, $\storb \Imp  \iwhile (\ea) \{\ca_1\} \ \Csmt{s}\ \storb \Imp \ca_1 ;  \iwhile (\ea) \{\ca_1\}$.
We have $\store' \approx_{\typenv,\typop} \storb'$ because $\store = \store'$ and $\storb = \storb'$.
We conclude that both configurations are equivalent, that is $\store' \Imp \ca' \approx_{\typenv,\typop} \storb' \Imp \cc'$.
The other cases  are treated similarly.\qed
\end{proof}

Sequential non-interference can be adapted to multi-threaded programs.
For that purpose, we extend the equivalence $\approx_{\typenv,\typop} $ to multi-threaded programs by:
\begin{itemize}
\item If $\forall x \in dom(M)=dom(M'),\ M(x) \approx_{\typenv,\typop} M'(x)$ then $M \approx_{\typenv,\typop} M'$
\item If $M \approx_{\typenv,\typop} M'$ and $\store \approx_{\typenv,\typop} \storb$ then $\store \Imp M  \approx_{\typenv,\typop} \storb \Imp M'$
\end{itemize}

\begin{theorem}[Concurrent Non-interference] \label{CNI}
Assume that $\typenv$ is a variable typing environment, that $\typop$ is a safe operator typing environment such that  
$M$ is well-typed. 
Assume also that $\store \Imp M_1 \approx_{\typenv,\typop} \storb \Imp M_2$.
Then, if $\store \Imp M_1 \Csmt{g} \store' \Imp M'_1$ then there are $\storb'$ and $M'_2$ s.t. $\storb \Imp M_2 \Csmt{g}^* \storb' \Imp M'_2$ and $\store' \Imp M_1 \approx_{\typenv,\typop} \storb' \Imp M_2 $.
\end{theorem}

\begin{proof}
Consequence of Theorem~\ref{thm:soundness}.
\qed
\end{proof}

\section{Sequential and concurrent temporal non-interferences}\label{sec:TNI}
Now we establish a property named temporal non-interference. 
This property ensures that the length of while-loops does not depend on variables of tier $\tiera$, and depends only on tier $\tierb$ variables. 
Consequently, a change in the value of a variable of tier $\tiera$ does not affect loop lengths.

For this, we define a loop length measure in Figure~\ref{fig:timeCom} based on the small step semantics of Figure~\ref{fig:Com}.
$\storb \Imp_0 \ca \Csmt{s}^* \storb' \Imp_{t} \ca'$ holds if $t$ is the number of while-loops, which are unfolded to reach $\storb' \Imp \ca'$ from $\storb \Imp \ca$, that is $t$ is the number of applications of the rule $\textit{(TW}_{\vrai}\textit{)}$ in a computation.
It is convenient to define the relation $\Cmpt{t}$ by
$\storb \Imp \ca \Cmpt{t} \storb' \Imp \ca'$ iff $\storb \Imp_0 \ca \Csmt{s}^* \storb' \Imp_{t} \ca'$.

\begin{figure*}[t]
\hrulefill\\
$$
\begin{array}{c} 
\infer{\store \Imp \ea \Csmt{e} \da }
{\store \Imp_t \xa \iasg \ea \Csmt{s} \store[\xa \leftarrow \da]}
\quad 
\quad
\infer{\phantom{\store \Imp \ea \Csmt{e} \da } }{\store \Imp_t \iskip \Csmt{s} \store} 
\quad
\quad
 \infer{\store \Imp_t \ca_1 \Csmt{s} \store_1}  
 {\store \Imp_t  \ca_1 \isep \ca_2 \Csmt{s} \store_1 \Imp_t \ca_2} 
 \\[15pt]
 \infer{\store \Imp_t \ca_1 \Csmt{s}  \store_1 \Imp_{t'} \ca_1' } 
{\store \Imp_t \ca_1 \isep \ca_2 \Csmt{s} \store_1 \Imp_{t'} \ca_1' ; \ca_2}
\quad
\quad
\infer{\store \Imp \ea \Csmt{e} w,\ w\in \{\vrai,\faux\}} 
{\store \Imp_{t} \iif \ea \ithen \ca_{\vrai} \ielse \ca_{\faux} \ \Csmt{s} \ \store \Imp_t \ca_{w}} 
\\[15pt]
\infer{\store \Imp \ea \Csmt{e} \faux} 
{\store \Imp_t \iwhile (\ea) \{\ca\} \ \Csmt{s}  \ \store}
\quad 
\quad
\infer{\store \Imp \ea \Csmt{e} \vrai }
{\store \Imp_t  \iwhile (\ea) \{\ca\} \ \Csmt{s} \ \store \Imp_{t+1} \ca ;  \iwhile (\ea) \{\ca\}}\ \textit{(TW$_{\vrai}$)}
\\[15pt]
\infer{M(\threada)= \ca \quad \store \Imp_0 \ca \Csmt{s} \store'  }
{\store \Imp_t M \ \Csmt{g} \ \store' \Imp_t M-\threada}
\quad 
\quad
\infer{M(\threada)= \ca \quad \store \Imp_t \ca \Csmt{s} \store' \Imp_{t'} \ca' }
{\store \Imp_t  M \ \Csmt{g} \ \store' \Imp_{t'} M[\threada:=\ca']}
\end{array}
$$
\caption{Loop length measure for commands and multi-thread programs}
\label{fig:timeCom}
\hrulefill
\end{figure*}

\begin{remark}\label{decoupe}
If $\typenv,\typop \imp \ca:\tiera$ and $\storb \Imp \ca \Csmt{s}^* \storb' \Imp \ca'$ then $\storb \Imp \ca \Cmpt{0} \storb' \Imp \ca'$ since there is no while loop inside $\ca$,
by Lemma~\ref{lem:confinement}.
Moreover, if $\storb \Imp \ca \Cmpt{t} \storb' \Imp \ca'$, then for every $k \leq t$ there are $\storb''$ and $\ca''$ such that 
$\storb \Imp \ca \Cmpt{k} \storb'' \Imp \ca'' \Cmpt{t-k} \storb' \Imp \ca' $.
\end{remark}

\begin{theorem}[Temporal non-interference] \label{thm:TNI}
Assume that $\typenv$ is a variable typing environment and $\typop$ is a safe operator typing environment 
s.t. $\typenv,\typop \imp \ca:\tya$ and $\typenv,\typop \imp \cc:\tya$. 
Assume also that $\store \Imp \ca \approx_{\typenv,\typop} \storb \Imp \cc$.
 Then, if $\store \Imp \ca \Cmpt{t} \store' \Imp \ca'$ then there are $\storb'$  and $\cc'$ s.t. $\storb \Imp \cc \Cmpt{t} \storb' \Imp \cc'$ and $\store' \Imp \ca'  \approx_{\typenv,\typop} \storb' \Imp \cc'$.
\end{theorem}

\begin{proof}
The proof goes by induction on $t$.
Suppose that $t=0$. This means that no rule \textit{(TW$_{\vrai}$)}  has been fired. 
The conclusion is a consequence of sequential non-interference Theorem~\ref{thm:soundness}.

Next, suppose that $\store \Imp \ca \Cmpt{t+1} \store' \Imp \ca'$. This means that a rule \textit{(TW$_{\vrai}$)} has been applied.
So suppose that $\ca =  \iwhile (\ea) \{\ca_1\}$ and that $\store \Imp \ea  \Csmt{e}  \vrai$. 
First, $\store \approx_{\typenv,\typop} \storb$ and Lemma~\ref{lem:simpSec} imply that $\storb \Imp \ea  \Csmt{e}  \vrai$.
Second,  since $\ca \approx_{\typenv,\typop} \cc$, we have $\ca = \cc$, by definition of $\approx_{\typenv,\typop}$.
Since $\ca' = \ca_1;\ca$, we have $\cc'=\ca_1;\ca$. 
Thus, $\ca'   \approx_{\typenv,\typop} \cc'$ and  $\storb\Imp \cc \Cmpt{t+1} \storb' \Imp \cc'$ hold. 
Moreover, we have $\store'=\store$ and $\storb'=\storb$, which implies that $\store' \approx_{\typenv,\typop} \storb'$.
We conclude that  $\store' \Imp \ca'  \approx_{\typenv,\typop} \storb' \Imp \cc'$.\\
The other cases are similar.
\qed
\end{proof}

We extend the relation $\Cmpt{t}$ as follows:
$\store \Imp  M \Cmpt{t} \store' \Imp M'$ if and only if 
$\store \Imp_0  M \ \Csmt{g}^*\ \store' \Imp_{t} M'$. As a corollary, we obtain a temporal non-interference result for multi-threaded programs. 

\begin{theorem}[Concurrent temporal non-interference] \label{CTNI}
Assume $\typenv$ is a variable typing environment and $\typop$ is a safe operator typing environment s.t. 
$M$ and $N$ are well typed. 
Assume that $\store \Imp M \approx_{\typenv,\typop} \storb \Imp N$.
Then, if $\store \Imp M\Cmpt{t} \store' \Imp M'$ then there are $\storb'$ and $N'$
s.t. $\storb \Imp N \Cmpt{t} \storb' \Imp N'$ and $\store' \Imp M' \approx_{\typenv,\typop} \storb' \Imp N'$.
\end{theorem}

\begin{proof}
Consequence of Theorem~\ref{thm:TNI}.
\qed
\end{proof}

\section{Multi threaded program running time}
An important point is that the number of tier $\tierb$ configurations in a computation is polynomially bounded in the size of tier $\tierb$ initial values.
\begin{lemma}\label{lem:espaceCalcul}
Let $M$ be a safe multi-threaded program wrt environments $\typenv$ and $\typop$.  If $\store \Imp  M \Cmpt{t} \store' \Imp M'$ then $\forall \xa \in \FV(M)$ such that $\typenv(\xa)=\tierb$ either $\store'(\xa) \in \{\emph{\vrai},\emph{\faux}\}$ or $\exists \xb \in \FV(M)$ such that $\typenv(\xb)=\tierb$ and $\store'(\xa)  \dord \store(\xb)$. 
\end{lemma}
\begin{proof}
Take one global computational step $\store \Imp  M \Csmt{\textit{g}} \store' \Imp M'$. 
Let $\xa$ be a variable assigned to in $M(\threada)$, for some thread identifier $\threada$, such that $\typenv(\xa) = \tierb$. 
$X$ can only be assigned to an expression $\ea$ of tier $\tierb$. 
By simple security lemma~\ref{lem:simpSec}, $\ea$ only contains neutral operators.
It means that either $\store'(\xa)$ is a truth value (corresponding to the computation of a predicate) or a subterm of a value of a tier $\tierb$ variable. 
\qed
\end{proof}

In the case where a multi-threaded program strongly terminates (i.e. $M \Downarrow$), we now establish that for all thread interactions, the maximal length of while-loops is polynomially bounded  in the size of tier $\tierb$ values of the initial store. This is a consequence of the temporal non-interference property. For this, define 
$\|\!-\!\|_\tierb$ by $\|\store\|_\tierb =\max_{\typenv(\xa)=\tierb}\taille{\store(\xa)}$. 

\begin{theorem}\label{thm:multiPtime}
Let $M$ be a safe multi-threaded program such that $M \Downarrow$. There is a polynomial $T$ such that for all stores $\store$, if 
$\store \Imp M \Cmpt{t} \store' \Imp M'$
then $t \leq T(\|\store\|_\tierb)$.
\end{theorem}

\begin{proof}
 By Theorem~\ref{CTNI}, the length of while-loops depends only on variables of tier $\tierb$.
 It implies that if we enter twice into a configuration with the same thread, say $\threada$, and the same values of tier $\tierb$, we know that $M$ is non-terminating. Indeed, it is possible to repeat the same transition again up to infinity by always firing the same sequence of global transitions.
 This contradicts the fact that $M \Downarrow$. 
 Consequently, we never enter twice in the same thread configuration. 
 Since the number of sub-words of a word of size $n$ is bounded by $n^2$, 
Lemma~\ref{lem:espaceCalcul} impies the number of distinct stores $\storb$ 
reachable from $\store$ is bounded polynomially by $\|\store\|_\tierb$.
It follows the number of configurations is polynomially bounded. 
Consequently there exists a polynomial $T$ 
such that the length of each terminating multi-threaded computation starting from $\store$ is bounded by $T(\|\store\|_\tierb)$.
Finally, we have that $t \leq T(\|\store\|_\tierb)$.
\qed
\end{proof}

We can now state our first main result:
\begin{theorem}\label{thm:Main}
Assume  that $M$ is a safe multi-threaded program. Moreover suppose that $M$ strongly terminates. There is a polynomial $Q$ such that:
$$\forall \da_1,\ldots,\da_n \in \Dom, \ \temps{M}{\da_1,\ldots,\da_n} \leq Q(\max_{i =1,n}(\taille{\da_i}))$$
\end{theorem}

\begin{proof}
Suppose that  $\store_0[\xaa \leftarrow \daa,\ldots,\xan \leftarrow \dan] \Imp M  \Cmpt{t}  \store' \Imp \emptyset $.
The overall computational time is bounded by 
$\temps{M}{\da_1,\ldots,\da_n} \leq r.t+r$, for some constant $r$ which depends on the size of $M$.
(Note that commands of tier $\tiera$ are computable in constant size.)
We conclude by  Theorem~\ref{thm:multiPtime} and by setting $Q(X)=r.T(X)+r$.\qed
\end{proof}

\section{A characterization of polynomial time functions}

We now come to a characterization of the set of functions computable in polynomial time.
A sequential program $M$ consists in a single thread program (i.e. $\dom(M)=\{\threada\}$) and an output variable, say $\xb$.
The partial function $\sem{M}$ computed by $M$ is then defined by: \\
$
\sem{M}(\da_1,\ldots,\da_n)  = w  \text{ iff } \store_0[\xaa\!\leftarrow\!\daa,\ldots,\xan\!\leftarrow\!\dan] \Imp M \Csmt{g}^* \store \Imp \emptyset \text{ and } \store(\xb)=w
$

\begin{theorem}\label{erst}
The set of functions computed by strongly terminating and safe sequential programs whose operators compute polynomial time functions is exactly $\ptime$, which is the set of polynomial time computable functions.
\end{theorem}

\begin{proof}
The polynomial runtime upper bound is a consequence of Theorem~\ref{thm:multiPtime}. The converse is a straightforward simulation of polynomial time Turing machines. The proof is postponed in Appendix.
\end{proof}

%
%
\section{Deterministic scheduling}\label{sec:ds}
Actually, we can extend our results to a class of deterministic schedulers. 
Till now, we have considered a non-deterministic scheduling policy but in return we require that multi-threaded programs strongly terminate.  
Define $\restrict{\store}$ as the restriction of the store $\store$ to tier $\tierb$ variables. 
Say that a deterministic scheduler $\sch$ is \emph{quiet} if the scheduling policy depends only on the current multi-threaded program $M$ and on $\restrict{\store}$. 
For example, a deterministic scheduler whose policy just depends on running threads, is quiet. 
Notice that $\storb \approx_{\typenv,\typop} \storb'$ iff $\restrict{\storb}=\restrict{\storb'}$.
Next, we replace the non-deterministic global transition of Figure~\ref{fig:Com} by:
$$
\begin{array}{c}
\infer{\sch(M,\restrict{\store}) =  \threada  \quad \store \Imp M(\threada) \Csmt{s} \store'  }
{\store \Imp M \ \Csmt{g} \ \store' \Imp M-\threada}
\qquad
\infer{\sch(M,\restrict{\store}) =  \threada \quad \store \Imp M(\threada) \Csmt{s} \store' \Imp \ca' }
{\store \Imp  M \ \Csmt{g} \ \store' \Imp M[\threada:=\ca']}
\end{array}
$$

\begin{theorem}
Let $M$ be a safe multi-threaded program s.t. $M$ is terminating wrt a deterministic and quiet scheduler $\sch$. 
There is a polynomial $Q$ such that:
$$\forall \da_1,\ldots,\da_n \in \Dom, \ \temps{M}{\da_1,\ldots,\da_n} \leq Q(\max_{i =1,n}(\taille{\da_i}))$$
\end{theorem}

\begin{proof}
The proof follows the outline of proofs of theorems~\ref{thm:multiPtime} and ~\ref{thm:Main}. 
Let $\store$ be the initial store, \ie $\store(\xai)=\dai$ for $i=1,n$ and $\store(\xai)=\motVide$ for $i>n$.
Since the computation of $\store \Imp M$ terminates wrt $\sch$, the temporal non-interference theorem~\ref{thm:TNI} implies 
that a loop can not reach  the configurations $\storb \Imp N$ and $\storb' \Imp N$ 
where their restrictions to tier $\tierb$ values are identical. That is $\restrict{\storb}=\restrict{\storb'}$.
Now, define $\config = \{ (\restrict{\storb},N) \ | \ \store \Imp M  \Csmt{g}^*  \storb \Imp N  \} $. 
The total length of loops is bounded by the cardinality of $\config$.
Following lemma~\ref{lem:espaceCalcul}, the cardinality of $\config$ is bounded by a polynomial in $\|\store\|_\tierb$.
As a result, the runtime of $\store \Imp M$ is bounded bounded by $Q(\max_{i =1,n}(\taille{\da_i}))$ for some polynomial $Q$.
\end{proof}

\bibliographystyle{plain}
\bibliography{icalp}

\begin{thebibliography}{10}

\bibitem{AD07}
R.~M. Amadio and F.~Dabrowski.
\newblock Feasible reactivity in a synchronous pi-calculus.
\newblock In {\em PPDP}, pages 221--230, 2007.

\bibitem{ADZ04}
R.~M. Amadio and S.~Dal-Zilio.
\newblock Resource control for synchronous cooperative threads.
\newblock In {\em CONCUR}, pages 68--82, 2004.

\bibitem{Baillot}
Patrick Baillot and Kazushige Terui.
\newblock Light types for polynomial time computation in lambda-calculus.
\newblock In {\em LICS}, IEEE Computer Society Press, pages 266--275, 2004.

\bibitem{BellLapadula76}
D.~E. Bell and L.J.~La Padula.
\newblock Secure computer system: unified exposition and multics
  interpretation.
\newblock Technical report, Mitre corp Rep., 1976.

\bibitem{BC92}
S.~Bellantoni and S.~Cook.
\newblock A new recursion-theoretic characterization of the poly-time
  functions.
\newblock {\em Computational Complexity}, 2:97--110, 1992.

\bibitem{Biba77}
K.~Biba.
\newblock Integrity considerations for secure computer systems.
\newblock Technical report, Mitre corp Rep., 1977.

\bibitem{BMM11}
G.~Bonfante, J.Y. Marion, and J.Y. Moyen.
\newblock Quasi-interpretations a way to control resources.
\newblock {\em Theo. Comput. Sci.}, 2011.

\bibitem{Boudol}
I.~Castellani and G.~Boudol.
\newblock Non-interference for concurrent programs.
\newblock In {\em ICALP}, volume 2076 of {\em Lecture Notes in Computer
  Science}, pages 382--395, 2001.

\bibitem{CPR07}
B.~Cook, A.~Podelski, and A.~Rybalchenko.
\newblock Proving thread termination.
\newblock In {\em PLDI}, pages 320--330, 2007.

\bibitem{Girard98}
J.-Y. Girard.
\newblock Light linear logic.
\newblock {\em Inf. Comput.}, 143(2):175--204, 1998.

\bibitem{Hofmann03}
M.~Hofmann.
\newblock Linear types and non-size-increasing polynomial time computation.
\newblock {\em Inf. Comput.}, 183(1):57--85, 2003.

\bibitem{Jones01}
N.~Jones.
\newblock The expressive power of higher-order types or, life without cons.
\newblock {\em J. Funct. Program.}, 11(1):5--94, 2001.

\bibitem{JonesK09}
N.~Jones and L.~Kristiansen.
\newblock A flow calculus of {\it wp}-bounds for complexity analysis.
\newblock {\em ACM Trans. Comput. Log.}, 10(4), 2009.

\bibitem{JonesCC}
N.D. Jones.
\newblock {\em Computability and complexity, from a programming perspective}.
\newblock MIT press, 1997.

\bibitem{Leivant94}
D.~Leivant.
\newblock A foundational delineation of poly-time.
\newblock {\em Inf. Comput.}, 110(2):391--420, 1994.

\bibitem{Lei-predicativeI}
D.~Leivant.
\newblock Predicative recurrence and computational complexity i: Word
  recurrence and poly-time.
\newblock In Peter Clote and Jeffery Remmel, editors, {\em Feasible
  Mathematrics II}. 1994.

\bibitem{MA11}
A.~Madet and R.~M. Amadio.
\newblock An elementary affine lambda-calculus with multithreading and side
  effects.
\newblock In {\em TLCA}, pages 138--152, 2011.

\bibitem{MarionLICS}
J.-Y. Marion.
\newblock A type system for complexity flow analysis.
\newblock In {\em LICS}, pages 123--132, 2011.

\bibitem{MP09}
J.Y. Marion and R.~P{\'e}choux.
\newblock Sup-interpretations, a semantic method for static analysis of program
  resources.
\newblock {\em ACM TOCL}, 10(4):27, 2009.

\bibitem{NigglW06}
K.-H. Niggl and H.~Wunderlich.
\newblock Certifying polynomial time and linear/polynomial space for imperative
  programs.
\newblock {\em SIAM J. Comput.}, 35(5):1122--1147, 2006.

\bibitem{SabelfeldMyersJSAC}
A.~Sabelfeld and A.~C. Myers.
\newblock Language-based information-flow security.
\newblock {\em IEEE J. Selected Areas in Communications}, 21(1):5--19, January
  2003.

\bibitem{SV98}
G.~Smith and D.~Volpano.
\newblock Secure information flow in a multi-threaded imperative language.
\newblock In {\em POPL}, pages 355--364. ACM, 1998.

\bibitem{VolpanoIS96}
D.~Volpano, C.~Irvine, and G.~Smith.
\newblock A sound type system for secure flow analysis.
\newblock {\em Journal of Computer Security}, 4(2/3):167--188, 1996.

\end{thebibliography}

\newpage
\appendix
\section{Appendix}
\subsection{Proofs}

\subsubsection{Characterization of polynomial time functions}

\setcounter{theorem}{6}
\begin{theorem}\label{erst}
The set of functions computed by strongly terminating and safe sequential programs whose operators compute polynomial time functions is exactly $\ptime$, which is the set of polynomial time computable functions.
\end{theorem}

\begin{proof}
By Theorem~\ref{thm:Main}, the execution time of a safe and strongly terminating sequential program is bounded by a polynomial in the size of the initial values. \\
In the other direction, we show that every polynomial time function over the set of words $\Dom$ can be computed by a safe and terminating program. Consider a Turing Machine $TM$, with one tape and one head,  which computes within $n^k$ steps for some constant $k$ and where $n$ is the input size. The tape of $TM$ is represented by two variables $\gauche$ and $\droite$ which contain respectively the reversed left side of the tape and the right side of the tape. States are encoded by constant words and the current state is stored in the variable $\vetat$. We assign to each of these three variables that hold a configuration of TM the tier $\tiera$.
A one step transition is simulated by a finite cascade of if-commands of the form:
\begin{lstlisting}[frame=none]
$\iif \meq{a}(\droite^\tiera)^\tiera$
   $\ithen$
      $\iif \meq{s}(\vetat^\tiera)^\tiera$ 
      $\ithen $
         $\vetat^\tiera \iasg s'^\tiera; : \tiera$ 
         $\gauche^\tiera \iasg \msuc{b}(\gauche^\tiera); : \tiera$
         $\droite^\tiera \iasg \mpred(\droite^\tiera) : \tiera$
      $\ielse \ldots : \tiera$
   $\ldots$
\end{lstlisting}
The above command expresses that if the current read letter is $a$ and the state is $s$, then the next state is $s'$, the head moves to the right and the read letter is replaced by $b$.
Since each variable inside the above command is of type $\tiera$, the type of the if-command is also $\tiera$. Moreover, since $\msuc{b}$ is a positive operator, its type is forced to be $\tiera \to \tiera$. $\meq{a}, \meq{s}$ and $\mpred$ being neutral operators, they can also be typed by $\tiera \to \tiera$.\\
Finally, it just remains to show that every polynomial can be simulated by a safe program of tier $\tierb$. We have already provided the programs for addition and multiplication in Example~\ref{e1} and we let the reader check that it can be generalized to any polynomial.\qed
\end{proof}

\subsection{Examples}
In what follows, let $\ea^\alpha$, respectively $\ca : \alpha$, be a notation meaning  that the expression $\ea$, respectively command $\ca$, is of type $\alpha$ under the considered typing environments.

\begin{example}\label{e1}
Consider the sequential programs $add_\xb$ and $mul_\xc$ that compute respectively addition and multiplication on unary words using the positive successor operator $+1$, in infix notation, and two neutral operators, $-1$  and
a unary predicate $>0$, both in infix notation.  
Both programs are safe by checking that their main commands are well-typed wrt the safe operator typing environment $\typop$ defined by $\typop(+1)=\{ \tiera \to \tiera\}$ and  $\typop(-1)=\typop(>0)=\{ \tierb \to \tierb\}$.
\begin{lstlisting}
$add_{\xb} :$                          $mul_\xc :$
   $\iwhile (\xa^\tierb > 0)^\tierb \{ $                           $\xc^\tiera \iasg 0^\tiera; :{\tiera}$
      $\xa^\tierb \iasg \xa^\tierb-1; : \tierb$                        $\iwhile ( \xa^{\tierb}> 0) ^\tierb\{$
      $\xb^\tiera \iasg \xb^\tiera+1 :\tiera$                            $\xa^\tierb \iasg \xa^\tierb-1; : {\tierb}$
   $ \}  :{\tierb}$                               $\xd^\tierb \iasg \xb^\tierb; : \tierb$
                                  $ \iwhile (\xb^{\tierb}> 0)^\tierb\{$ 
                                    $\xb^\tierb \iasg \xb^\tierb-1; : {\tierb}$
                                    $\xc^\tiera \iasg \xc^\tiera+1 : {\tiera}$
                                  $\}; : {\tierb}$   
                                  $\xb^\tierb \iasg \xd^\tierb : \tierb$ 
                               $\} : {\tierb}$   
\end{lstlisting}
\end{example}

\begin{example}
Consider the following multi-thread $M$ composed of two threads $x$ and $y$ computing on unary numbers:
\begin{lstlisting}
$x:$                          $y:$
   while $(\xa^\tierb > 0)^\tierb \{ $                      while $(\xb^\tierb > 0)^\tierb \{ $   
      $\xc^\tiera \iasg \xc^\tiera+1; : \tiera$                            $\xc^\tiera =0; : \tiera $
      $\xa^\tierb \iasg \xa^\tierb-1; : \tierb$                            $\xb^\tierb \iasg \xb^\tierb-1; : \tierb$
   $ \}  :{\tierb}$                            $ \}  :{\tierb}$   
\end{lstlisting}
This program is strongly terminating. Moreover, given a store $\store$ such that $\store(\xa)=n$ and $\store(\xc)=0$, if $\store \Imp M \Csmt{g}^k \store' \Imp \emptyset$ then $\store'(\xc) \in [0,n]$. $M$ is safe using an operator typing environment $\typop$ such that $\typop(-1)=\typop(>0)=\{\tierb \to \tierb\}$ and $\typop(+1)=\{\tiera \to \tiera\}$ and $M \Downarrow$. Consequently, by Theorem~\ref{thm:multiPtime}, there is a polynomial $T$ such that for each store $\store$, $k \leq T( \| \store \|_\tierb)$. 
\end{example}

\begin{example}
Consider the following multi-thread $M$ that shuffles two strings given as inputs:
\begin{lstlisting}
$x:$                          $y:$
   while $(\neg \meq{\epsilon}(\xa^\tierb) )^\tierb \{ $                      while $(\neg \meq{\epsilon}(\xb^\tierb ) )^\tierb \{ $   
      $\xc^\tiera \iasg concat(head(\xa^\tierb),\xc^\tiera); : \tiera$                            $\xc^\tiera \iasg concat(head(\xb^\tierb),\xc^\tiera); : \tiera$
      $\xa^\tierb \iasg pred(\xa^\tierb); : \tierb$                            $\xb^\tierb \iasg pred(\xb^\tierb); : \tierb$
   $ \}  :{\tierb}$                            $ \}  :{\tierb}$   
\end{lstlisting}
The negation operator $\neg$ and $\meq{\epsilon}$ are unary predicates and consequently can be typed by $\tierb \to \tierb$. The operator $head$ returns the first symbol of a string given as input and can be typed by $\tierb \to \tiera$ since it is neutral. The $\mpred$ operator can typed by $\tierb \to \tierb$ since its computation is a subterm of the input. Finally, the concat operator that performs the concatenation of the symbol given as first argument with the second argument can be typed by $\tiera \to \tiera \to \tiera$ since $\taille{\sem{concat}(u,v)}=\taille{v}+1$.
This program is safe and strongly terminating consequently it also terminates in polynomial time.
\end{example}

\begin{example}
Consider the following multi-thread $M$:
\begin{lstlisting}
$x:$                          $y:$
   while $(\xa^\tierb >0)^\tierb \{ $                      while $(\xb^\tierb >0)^\tierb \{ $   
      $\xb^\tierb \iasg \xa^\tierb; : \tierb$                            $\xc^\tiera \iasg \xc^\tiera+1; : \tiera$
      $\xa^\tierb \iasg \xa^\tierb-1; : \tierb$                            $\xb^\tierb \iasg  \xb^\tierb-1; : \tierb$
   $ \}  :{\tierb}$                            $ \}  :{\tierb}$   
\end{lstlisting}
Observe that, contrarily to previous examples, the guard of $y$ depends on information flowing from $X$ to $Y$.
Given a store $\store$ such that $\store(\xa)=n$, $\store(\xb)=\store(\xc)=0$, if $\store \Imp M \Csmt{g}^k \store' \Imp \emptyset$ then $\store'(\xc) \in [0,n \times (n+1)/2]$. 
This multi-thread is safe with respect to a safe typing operator environment $\typop$ such that $\typop(-1)=\typop(>0)=\{\tierb \to \tierb\}$ and $\typop(+1)=\{\tiera \to \tiera\}$. Moreover it strongly terminates. Consequently, it also terminates in polynomial time.
\end{example}

\begin{example}
The following program computes the exponential:
\begin{lstlisting}[frame=none]
$exp_{\xb}(\xa^{\tierb},\xb^{\tiera}) :$
   $\iwhile ( \xa^{\tierb}> 0)\{$
      $\xd^? \iasg \xb^\tiera; : {?}$   
      $ \iwhile (\xd^{?}> 0)\{$ 
         $\xb^\tiera \iasg \xb^\tiera+1; : {\tiera}$
         $\xd^? \iasg \xd^?-1 : {?}$
      $\}; : {\tierb}$ 
      $\xa^\tierb \iasg \xa^\tierb-1 : \tierb$
   $\}; : {\tierb}$
\end{lstlisting}
It is not typable in our formalism. Indeed, suppose that it is typable. The command $\xb \iasg \xb+1$ enforces $\xb$ to be of tier $\tiera$ since $+1$ is positive. Consequently, the command $\xd \iasg \xb$ enforces $\xd$ to be of tier $\tiera$ because of typing discipline for assignments. However, the innermost while loop enforces $\xd >0$ to be of tier $\tierb$, so that $\xd$ has to be of tier $\tierb$ (because $\tiera \to \tierb$ is not permitted for a safe operator typing environment) and we obtain a contradiction.
\end{example}

\begin{example}
As another counter-example, consider now the addition $badd$ on binary words:
\begin{lstlisting}
$badd_{\xb} :$                         
   $\iwhile (\xa^? > 0)^? \{ $                   
      $\xa^? \iasg \xa^?-1; : \ ?$                     
      $\xb^\tiera \iasg \xb^\tiera+1 :\tiera$                     
   $ \}  :{\tierb}$                        
\end{lstlisting}
Contrarily to Example~\ref{e1}, the above program is not typable because the operator $-1$ has now type $\typop(-1)=\{ \tiera \to \tiera\}$. Indeed it cannot be neutral since binary predecessor is not a subterm operator. Consequently, $-1$ is positive and the assignment $\xa \iasg \xa-1$ enforces $\xa$ to be of type $\tiera$ whereas the loop guard enforces $\xa$ to be of tier $\tierb$. Note that this counter-example is not that surprising in the sense that a binary word of size $n$ may lead to a loop of length $2^n$ using the $-1$ operator. Of course this does not imply that the considered typing discipline rejects computations on binary words, it only means that this type system rejects exponential time programs. Consequently,``natural'' binary addition algorithms are captured as illustrated by the following program that computes the binary addition on reversed binary words of equal size:
\begin{lstlisting}
$binary$_$add_{\xc} :$   
   $\iwhile (\neg \meq{\epsilon}(\xa^\tierb)  )^\tierb \{ $  
      $\xe^\tiera \iasg result(bit(\xa^\tierb),bit(\xb^\tierb),bit(\xf^\tierb)); : \tiera$
      $\xf^\tierb \iasg carry(bit(\xa^\tierb),bit(\xb^\tierb),bit(\xf^\tierb)); : \tierb$
      $\xc^\tiera \iasg concat(\xe^\tiera, \xc^\tiera); : \tiera$
      $\xa^\tierb \iasg \mpred{}(\xa^\tierb); : \tierb$
      $\xb^\tierb \iasg \mpred{}(\xb^\tierb); : \tierb$              
   $ \}  :{\tierb}$       
\end{lstlisting}
As usual, $\mpred$ is typed by $\tierb \to \tierb$. The negation operator $\neg$ and $\meq{\epsilon}$ are predicates and, consequently, can be typed by $\tierb \to \tierb$, since they are neutral. The operator $bit$ returns $\vrai$ or $\faux$ depending on whether the word given as input has first digit $1$ or $0$, respectively. Consequently, it can be typed by $\tierb \to \tierb$. The operators $carry$ and $result$, that compute the carry and the result of bit addition, can be typed by $\tierb \to \tierb \to \tierb \to \tierb$ since they are neutral. Finally, the operator $concat(x,y)$ defined by if $\sem{bit}(x)=i, i \in\{0,1\}$ then $\sem{concat}(x,y)=i.y$ is typed by $\tiera \to \tiera \to \tiera$. Indeed it is a positive operator since $\taille{\sem{concat}(x,y)} = \taille{y}+1$.
\end{example}
\end{document}